\documentclass[envcountsame]{llncs}

\scrollmode
\usepackage{amsmath}
\usepackage{amsfonts}
\usepackage{amssymb}
\usepackage{latexsym}
\usepackage{stmaryrd}
\usepackage{array}
\usepackage{exscale}
\newcommand{\nc}{\newcommand}
\newcommand{\ol}{\overline}

\newcommand{\es}{\emptyset}
\newcommand{\sm}{\setminus}
\newcommand{\ve}{\varepsilon}

\newcommand{\bc}{\bigcup}
\newcommand{\bca}{\bigcap}
\newcommand{\Lra}{\Leftrightarrow}

\newcommand{\Ra}{\Rightarrow}

\newcommand{\ra}{\rightarrow}

\newcommand{\sse}{\subseteq}

\newcommand{\spe}{\supseteq}
\newcommand{\fa}{\forall}

\newcommand{\mr}{\mathrm}
\newcommand{\mc}{\mathcal}
\newcommand{\mf}{\mathfrak}

\newcommand{\DMO}{\DeclareMathOperator}
\newcommand{\DST}{\displaystyle}

\newcommand{\NN}{\mathbb{N}}
\newcommand{\NNZ}{\NN_0}
\newcommand{\ZZ}{\mathbb{Z}}

\newcommand{\RR}{\mathbb{R}}

\mathchardef\breakingcomma\mathcode`\,
{\catcode`,=\active
  \gdef,{\breakingcomma\discretionary{}{}{}}
}

\usepackage{listings}
\newcommand{\inl}[1]{\lstinline$#1$}
\newcommand{\und}{{\:\wedge\:}} 
\newcommand{\mb}{{\:|\:}} 
\newcommand{\set}[1]{\{ #1 \}}

\nc{\simlvi}[1]{\!\sim_{#1}}
\DeclareMathOperator{\symdif}{\vartriangle} 
\DeclareMathOperator{\addcup}{{\stackrel{\text{\raisebox{-2.2ex}[-0ex][-0ex]{\large$\cdot$}}}{\cup}}} 
\nc{\apprel}[3]{{#1}(#2)_{(#3)}} 
\nc{\cmpli}[1]{\complement^1_{#1}} 
\nc{\cmplzi}[1]{\complement^0_{#1}} 
\nc{\cmplzoi}[1]{\complement^*_{#1}} 

\nc{\zf}{\mr{ZF}}
\nc{\zfmf}{\zf^0} 
\nc{\zfc}{\mr{ZFC}}
\nc{\zfcmf}{\zfc^0} 
\nc{\bst}{\mr{BST}} 
\newcommand{\nni}{\NNZ \cup \{+\infty\}} 
\newcommand{\tb}[2]{\set{#1, \dots, #2}} 
\providecommand{\abs}[1]{\lvert #1 \rvert} 
\providecommand{\norm}[1]{\lVert #1 \rVert} 
\makeatletter
\DeclareRobustCommand{\genericinterval}[2]{%
  \@ifstar{\genericinterval@star{#1}{#2}}{\genericinterval@nostar{#1}{#2}}}
\newcommand{\genericinterval@star}[4]{\mathopen{}\mathclose{\left#1#3,#4\right#2}}
\newcommand{\genericinterval@nostar}[4]{\mathopen{#1}#3,#4\mathclose{#2}}

\makeatother
\nc{\untit}[2]{{#1}^{#2 \downarrow}} 
\nc{\obit}[2]{{#1}^{#2 \uparrow}} 
\nc{\inzEKi}[1]{\mc{I}^{\mr{V}}_{#1}}

\nc{\inzKEi}[1]{\mc{I}^{\mr{E}}_{#1}}

\nc{\adjEi}[1]{\mc{A}^{\mr{V}}_{#1}}

\nc{\BD}[1]{{#1}\text{-}\mr{BD}}

\nc{\konv}[2]{{#1}[{#2}]} 
\nc{\actpres}[1]{\Phi_{#1}} 
\nc{\Prim}{\mc{PR}} 

\nc{\sselr}{\sse^{\mapsto}}
\nc{\sserl}{\sse^{\mapsfrom}}
\nc{\spelr}{\spe^{\mapsto}}
\nc{\sperl}{\spe^{\mapsfrom}}
\nc{\ball}[1]{\mr{B}^{#1}} 
\nc{\oball}[1]{\breve{\mr{B}}^{#1}} 
\nc{\pball}[1]{\dot{\mr{B}}^{#1}} 
\nc{\prr}[1]{\dot{\RR}^{#1}} 
\nc{\sph}[1]{\mr{S}^{#1}} 
\nc{\ssim}[1]{s\sigma_{#1}} 
\nc{\koerper}[1]{\norm{#1}}
\nc{\Ccovdim}{\mc{CD}}
\nc{\Cinddim}{\mc{SID}}

\nc{\CInddim}{\mc{LID}}

\DeclareMathOperator{\diffop}{D} 
\DeclareMathOperator*{\diffoplimit}{D} 
\nc{\diffopc}[1]{\sideset{_{#1}}{}\diffoplimit} 
\nc{\diffopp}[1]{\diffop_{#1}} 
\nc{\diffopcp}[2]{\sideset{_{#2}}{_{#1}}\diffoplimit} 
\nc{\meanH}[2]{\mf{M}_{#1,#2}} 
\nc{\emean}[2]{\mf{M}_{\exp_{#1},#2}} 
\DeclareMathOperator{\mor}{Mor}
\DeclareMathOperator{\Hom}{Hom} 
\nc{\autoerw}[1]{\mr{Aut}^{#1}} 
\nc{\komma}[2]{(#1 \downarrow #2)} 
\nc{\Kmat}{\mf{MAT}} 
\nc{\Khmat}{\mf{HMAT}} 
\nc{\homfun}[1]{\mor_{#1}(-_1,-_2)} 
\nc{\homfunae}[1]{\mor_{#1}(-_1)} 
\nc{\homfunbe}[1]{\mor_{#1}(-_2)} 
\nc{\homfunxy}[3]{\mor_{#1}(#2(-_1), #3(-_2))}
\nc{\homfunx}[2]{\mor_{#1}(#2(-_1), -_2)}
\nc{\homfuny}[2]{\mor_{#1}(-_1, #2(-_2))}
\nc{\homfuna}[2]{\mor_{#1}(#2, -)} 
\nc{\homfunb}[2]{\mor_{#1}(-, #2)} 
\nc{\hhomfuna}[2]{\Hom_{#1}(#2, -)} 
\nc{\hhomfunb}[2]{\Hom_{#1}(-, #2)} 
\newcommand{\Va}{\mc{V\hspace{-0.1em}A}}

\newcommand{\Lit}{\mc{LIT}}
\newcommand{\Cl}{\mc{CL}}
\newcommand{\Cls}{\mc{CLS}}

\newcommand{\Sat}{\mc{SAT}}

\newcommand{\Usat}{\mc{USAT}}

\newcommand{\Musat}{\mc{M\hspace{0.8pt}U}} 
\newcommand{\Musati}[1]{\Musat_{\!#1}} 
\newcommand{\Smusat}{\mc{S}\Musat} 
\newcommand{\Smusati}[1]{\Smusat_{\!#1}}

\nc{\Clsoo}{\Cls^{1,1}} 
\DeclareMathOperator{\lit}{lit}
\DeclareMathOperator{\var}{var}

\DMO{\dos}{ds} 
\DMO{\mdos}{mds} 
\newcommand{\Clash}{\mc{HIT}} 

\newcommand{\Uclash}{\mc{U}\Clash} 
\newcommand{\Uclashi}[1]{\Uclash_{\!\!#1}}

\DeclareMathOperator{\res}{\diamond} 
\DeclareMathOperator{\dpl}{DP} 
\newcommand{\dpi}[1]{\dpl_{\!#1}}
\DMO{\premr}{ax} 
\DMO{\concr}{C} 
\DMO{\allcr}{cl} 

\DMO{\thardness}{thd} 
\DMO{\phardness}{phd} 
\DMO{\whardness}{awid} 
\DMO{\dep}{dep} 
\DMO{\hts}{hs} 
\DMO{\semspace}{css} 
\DMO{\resspace}{crs} 
\DMO{\treespace}{cts} 
\nc{\bth}[1]{\langle{#1}\rangle} 
\DMO{\rsub}{r_S} 
\DMO{\rk}{r} 
\DMO{\ro}{\rk_1} 
\DMO{\rki}{\rk_{\infty}} 
\DMO{\rpl}{r^{pl}} 
\DMO{\ropl}{\rk_1^{pl}} 
\nc{\rslur}{\xrightarrow{\text{SLUR}}} 
\nc{\rslurs}{\rslur_{\!*}} 
\DMO{\slur}{slur} 
\nc{\Slur}{\mc{SLUR}} 
\nc{\rkslur}[1]{\xrightarrow{\text{SLUR}_{#1}}} 
\nc{\rkslurs}[1]{\rkslur{#1}_{\!*}} 
\nc{\Altsluri}[1]{\Slur(#1)}
\nc{\Altslurstari}[1]{\Slur\text{\textasteriskcentered}(#1)}
\nc{\Canoni}[1]{\mr{CANON}(#1)}
\nc{\rkslurstar}[1]{\xrightarrow{\text{SLUR\textasteriskcentered}#1}} 
\nc{\rkslursstar}[1]{\rkslurstar{#1}_{\!*}} 
\DMO{\slurstar}{\slur\!\text{\textasteriskcentered}}
\nc{\Urefc}{\mc{UC}}
\nc{\Propc}{\mc{PC}}
\nc{\Wrefc}{\mc{WC}} 
\DeclareMathOperator{\ldeg}{ld} 
\DeclareMathOperator{\vdeg}{vd} 
\DeclareMathOperator{\minvdeg}{\mu\!\vdeg} 
\DMO{\varmvd}{\var_{\minvdeg}} 
\DMO{\nfc}{fc} 
\DMO{\maxnfc}{\nu\!\nfc} 
\nc{\svbf}{\mc{VB}} 
\nc{\svbfs}{\mc{VB}^*} 
\DMO{\potp}{pp} 
\DMO{\potprec}{NM} 
\DMO{\minnonmer}{VDM} 
\DMO{\minnonmerh}{VDH} 
\DMO{\maxsmar}{FCM} 
\DMO{\maxsmarh}{FCH} 
\DMO{\varsing}{\var_s} 
\DMO{\varosing}{\var_{1s}} 
\DMO{\varnosing}{\var_{\neg1s}} 
\nc{\Musatns}{\Musat'} 
\nc{\Musatnsi}[1]{\Musati{#1}'}
\nc{\Smusatns}{\Smusat'} 
\nc{\Smusatnsi}[1]{\Smusati{#1}'}
\nc{\Uclashns}{\Uclash'} 
\nc{\Uclashnsi}[1]{\Uclashi{#1}'}
\nc{\tsdp}{\xrightarrow{\text{sDP}}}
\nc{\tsdps}{\tsdp_{\!*}}
\nc{\tosdp}{\xrightarrow{\text{1sDP}}}
\nc{\tosdps}{\tosdp_{\!*}}
\DMO{\sdp}{sDP} 
\DMO{\osdp}{sDP_1} 
\nc{\cflmusat}{\mc{CF}\Musat} 
\nc{\cflmusati}[1]{\mc{CF}\Musati{#1}}
\nc{\cflimusat}{\mc{CFI}\Musat} 
\DMO{\sNF}{sNF} 
\DMO{\eqp}{eqp} 
\DMO{\sgp}{sp} 
\DMO{\singind}{si} 
\DMO{\osingind}{si_1} 
\DMO{\shyp}{svh} 
\DMO{\sdph}{ssh} 
\DMO{\msdph}{mss} 
\DMO{\osdph}{ssh_1} 
\DMO{\mosdph}{mss_1} 
\DMO{\mps}{mps} 
\DMO{\purec}{puc} 
\DMO{\doping}{D}
\nc{\glue}[4]{\operatorname{glue}((#1,#2), (#3,#4))} 
\nc{\gluea}[3]{#1 \mathbin{\boxplus}_{#3} #2} 
\newcommand{\cor}{\mathbin{\ovee}} 
\DMO{\frl}{fl} 
\nc{\Con}{\mr{Con}}
\nc{\Log}{\mr{Log}}
\nc{\Lin}{\mr{Lin}}
\nc{\Pol}{\mr{Pol}}
\nc{\ExL}{\mr{ExL}}
\nc{\ExP}{\mr{ExP}}
\nc{\CTime}{\mr{CTime}}
\nc{\CSpace}{\mr{CSpace}}
\nc{\LTime}{\mr{LTime}}
\nc{\LSpace}{\mr{L}}
\nc{\NLSpace}{\mr{NL}}
\nc{\LinTime}{\mr{LinTime}}
\nc{\LinSpace}{\mr{LinSpace}}
\nc{\PTime}{\mr{P}}
\nc{\PSpace}{\mr{PSpace}}
\nc{\Np}{\mr{NP}}
\nc{\Conp}{\text{coNP}}
\nc{\NPSpace}{\mr{NPSpace}}
\nc{\CoNPSpace}{\mr{coNPSpace}}
\nc{\ELTime}{\mr{ELTime}}
\nc{\ELSpace}{\mr{ELSpace}}
\nc{\EPTime}{\mr{EPTime}}
\nc{\EPSpace}{\mr{EPSpace}}
\nc{\NEPTime}{\mr{NEPTime}}
\nc{\polydelta}[1]{\Delta_{#1}^{\mr P}}
\nc{\polypi}[1]{\Pi_{#1}^{\mr P}}
\nc{\polysigma}[1]{\Sigma_{#1}^{\mr P}}
\nc{\Ph}{\mr{PH}}

\nc{\Dp}{D^P}
\nc{\PllC}[2]{{\text{$\mr{PT}$/$\mr{WK}$}(#1, #2)}} 
\nc{\Nc}{\mr{NC}}
\nc{\Nci}[1]{\Nc^{#1}}
\nc{\Ac}{\mr{AC}}
\nc{\Aci}[1]{\Ac^{#1}}
\nc{\pmodpoly}{P / \mathrm{poly}}
\nc{\Wh}[1]{\mr{W}[#1]} 
\nc{\Rl}{\mr{RL}}
\nc{\coRl}{\mr{coRL}}
\nc{\Rp}{\mr{RP}}
\nc{\coRp}{\mr{coRP}}
\nc{\Zpp}{\mr{ZPP}}
\nc{\Bpp}{\mr{BPP}}
\nc{\Pp}{\mr{PP}}
\nc{\Reach}{\mr{STCON}} 
\nc{\Undreach}{\mr{USTCON}} 
\nc{\Pcol}[2]{\mr{COL}(#1,#2)} 
\nc{\Pscol}[2]{\mr{SCOL}(#1,#2)} 
\nc{\Psorcol}[2]{\mr{SORCOL}(#1,#2)} 
\DMO{\slp}{slp}
\nc{\Mss}{\mr{MSS}}
\nc{\Key}{\mr{KEY}}
\nc{\Keyi}[1]{\Key_{\!#1}}
\nc{\Nbmss}{N_{\mr{bm}}} 
\nc{\Nbkey}{N_{\mr{bk}}} 
\nc{\Rnb}{N_{\mr{b}}}
\nc{\Rnk}{N_{\mr{k}}}
\nc{\Rnr}{N_{\mr{r}}}

\nc{\Byte}{\mr{B}[8]}
\nc{\QByte}{\mr{B}[4,8]}
\nc{\KByte}{\mc{B}} 
\nc{\RQByte}{\mc{QB}} 

\nc{\ramz}[3]{\mr{ram}_{#1}^{#2}(#3)} 
\nc{\waez}[2]{\mr{vdw}_{#1}(#2)} 
\nc{\gtz}[2]{\mr{grt}_{#1}(#2)} 
\nc{\pdwaez}[2]{\mr{vdw}_{#1}^{\mr{pd}}(#2)} 
\nc{\absfeh}[1]{\delta_{#1}} 
\nc{\relfeh}[1]{\ve_{#1}} 
\usepackage{theorem} 
\usepackage{hyperref} 
\theorembodyfont{\rmfamily}

\theorembodyfont{}
\newcounter{dDef} 

\newcounter{dLem} 

\newcounter{dThm} 

\newcounter{dPro} 

\newcounter{Beispielzaehler}

\nc{\bm}{\boldmath}
\nc{\bmm}[1]{\mbox{\bm$\DST #1$}}
\nc{\mi}[1]{\bmm{\mathrm{(#1):}} \quad}
\usepackage[active]{srcltx}

\newcommand{\Schrift}{report}
\newcommand{\Liste}{minimal unsatisfiability \sep hitting clause-set \sep disjoint/orthogonal tautology \sep deficiency \sep Finiteness Conjecture \sep singular variables \sep full subsumption resolution \sep irreducible CNF \sep clause-factor}
\providecommand{\sep}{, }

\DMO{\maxnhitdef}{NV}
\nc{\Dt}[1]{\mc{F}_{#1}} 

\DMO{\isouhit}{NI}

\newcommand{\Irred}{\mc{IRD}} 

\DMO{\nsv}{\mathit{n}_s}
\DMO{\nosv}{\mathit{n}_{1s}}
\DMO{\nnosv}{\mathit{n}_{\neg1s}}

\newcommand{\Clpr}{\mc{CIR}} 
\newcommand{\Pruclash}{\mc{IUH}} 
\newcommand{\Pruclashi}[1]{\Pruclash_{#1}}
\newcommand{\Ipruclash}{\mc{NIUH}} 
\newcommand{\Ipruclashi}[1]{\Ipruclash_{#1}}

\begin{document}

\pagestyle{headings}

\title{Unsatisfiable hitting clause-sets with three more clauses than variables}

\author{\href{http://cs.swan.ac.uk/~csoliver}{Oliver Kullmann}\inst{1} \and \href{http://logic.sysu.edu.cn/2005/english/PEOPLE/200510/english_277.html}{Xishun Zhao}\thanks{Partially supported by NSFC Grant 61272059 and NSSFC Grant 13\&ZD186.}\inst{2}}
\institute{Swansea University \and Sun Yat-sen University, Guangzhou}

\maketitle

\begin{abstract}
  Hitting clause-sets (as CNFs), known in DNF language as ``disjoint'' or ``orthogonal'', are clause-sets $F$, such that any $C, D \in F$, $C \ne D$, have a literal $x \in C$ with $\ol{x} \in D$. The set of unsatisfiable such $F$ is denoted by $\Uclash \subset \Musat$ (minimal unsatisfiability). A basic fact is $\delta(F) \ge 1$ for $F \in \Musat$, where the deficiency $\delta(F) := c(F) - n(F)$ is the difference between the number of clauses and the number of variables. Via the known singular DP-reduction, generalising unit-clause propagation, every $F \in \Uclash$ can be reduced to its (unique) ``non-singular normal form'' $\sNF(F) \in \Uclashns$, where $\delta(\sNF(F)) = \delta(F)$, and $\Uclashns \subset \Uclash$ is the subset of non-singular elements, i.e., every variable occurs positively as well as negatively at least twice.

  The \emph{Finiteness Conjecture} (FC) is that for every $k \in \NN$ the number $n(F)$ of variables for $F \in \Uclashns$ with $\delta(F) = k$ is bounded. This conjecture is part of the project of classifying $\Uclashi{\delta=k}$. In this \Schrift{} we prove FC for $k = 3$ (known for $k \le 2$). For this, a central novel concept is transferred from number theory (Berger et al 1990 \cite{BergerFelzenbaumFraenkel1990Covers}), namely the fundamental notion of \emph{clause-irreducible clause-sets} $F$, having no non-trivial \emph{clause-factors} $F'$, which are $F' \sse F$ logically equivalent to some clause. The derived factorisations allow to reduce FC to the clause-irreducible case. Another new tool is \emph{nearly-full-subsumption resolution} (nfs-resolution), which allows to change certain pairs $C, D$ of clauses. Clause-sets which become clause-reducible after a series of nfs-resolutions are called \emph{nfs-reducible}, and we can furthermore reduce FC to the nfs-irreducible case.
\end{abstract}



\section{Introduction}
\label{sec:intro}

Disjoint or orthogonal DNFs (every two terms/conjuncts/cubes have a conflict) have been playing an important role for boolean functions and their applications from the beginning, exploiting that the tautology problem (and also the counting problem) is computable in polynomial time; see \cite[Section 1.6, Chapter 7]{CramaHammer2011BooleanFunctions} for some overview. As CNFs, they are more known as hitting clause-sets, denoted by $\Clash$, and one of their earliest use is \cite{Iw89} (for counting solutions; see \cite[Section 13.4.2]{SS09HBSAT} for an extension). In this \Schrift, we study the unsatisfiable elements of $\Clash$, denoted by $\Uclash$; see \cite[Section 11.4.2]{Kullmann2007HandbuchMU} for some basic information. Our main context is the study of minimally unsatisfiable clause-sets ($\Musat$; see \cite{Kullmann2007HandbuchMU}), which is organised in layers by the deficiency $\delta$, and where the central Finiteness Conjecture is that every such layer can be described by finitely many ``patterns''. For $\Uclash \subset \Musat$ this means that every layer contains only finitely many isomorphism types (after a basic reduction), and this is the main problem studied in this \Schrift. The basic definitions are as follows.

$\Clash$ is the set of clause-sets $F$, such that for all $C, D \in F$, $C \ne D$, there is $x \in C$ with $\ol{x} \in D$. The set of unsatisfiable hitting clause-sets, denoted by $\Uclash$, is the set of $F \in \Clash$ with $\sum_{C \in F} 2^{-\abs{C}} = 1$. As measures we use $c(F) := \abs{F} \in \NNZ$ for the number of clauses of $F$, and $n(F) := \abs{\var(F)} \in \NNZ$ for the number of variables of $F$, while the deficiency is defined as $\delta(F) := c(F) - n(F) \in \ZZ$. For $F \in \Uclash$ holds $\delta(F) \ge 1$ (an instructive exercise for the reader, or see \cite{Kullmann2007HandbuchMU}). Finally $\Uclashns \subset \Uclash$, the set of nonsingular $F \in \Uclash$, is given by the condition, that for every $v \in \var(F)$ there exist (at least) four different clauses $A,B,C,D \in F$ with $v \in A, B$ and $\ol{v} \in C, D$. A central problem of the field is the \emph{Finiteness Conjecture} (FC; Conjecture 25 in \cite{KullmannZhao2011Bounds}):

\begin{definition}\label{def:maxvarhitdef}
  $\bmm{\maxnhitdef(k)} \in \nni$ is the supremum of $n(F)$ for $F \in \Uclashnsi{\delta=k}$.
\end{definition}

\begin{conjecture}\label{con:basicn}
  For every $k \in \NN$ we have  $\maxnhitdef(k) < +\infty$.
\end{conjecture}

\begin{example}\label{exp:uhit12}
  By \cite{DDK98} we know $\maxnhitdef(1) = 0$ (via $\set{\bot}$). By \cite{KleineBuening2000SubclassesMU} up to isomorphism there are two elements in $\Uclashnsi{\delta=2}$: $\bmm{\Dt{2}} := \set{\set{1,2}, \set{-1,-2}, \set{-1,2}, \set{-2,1}}$ and $\bmm{\Dt{3}} := \set{\set{1,2,3}, \set{-1,-2,-3},\set{-1,2},\set{-2,3},\set{-3,1}}$. Thus $\maxnhitdef(2) = 3$,
\end{example}

Using $\set{C} \cor F := \set{C \cup D : D \in F}$ for clauses $C$ and clause-sets $F$ with $\var(C) \cap \var(F) = \es$, we obtain more examples with high $\maxnhitdef(k)$:
\begin{lemma}\label{lem:Nkinc}
  For $m \in \NN$ let $K_m$ be defined as follows: $K_1 := \Dt{3}$, while $K_{m+1}$ is obtained from $K_m$ by taking a copy $F'$ of $\Dt{3}$ with $\var(F') \cap \var(K_m) = \es$, take a new variable $v \notin \var(K_m) \cup \var(F')$, and let $K_{m+1} := (\set{\set{v}} \cor K_m) \cup (\set{\set{\ol{v}}} \cor F')$. Then $K_m \in \Uclashnsi{\delta=m+1}$ with $n(K_m) = 3 + (m-1) \cdot 4$. So we get $\maxnhitdef(k) \ge 3 + (k-2) \cdot 4 = 4 k - 5$ for $k \ge 2$.
\end{lemma}

We believe the $K_m$ have the maximal number of variables for deficiency $m+1$, and so we consider the following strengthening of Conjecture \ref{con:basicn}:
\begin{conjecture}\label{con:concfinhitstr}
  For $k \in \NN$, $k \ge 2$, we have $\maxnhitdef(k) = 4 k - 5$.
\end{conjecture}
The values of $k \mapsto 4k-5$ for $2 \le k \le 6$ are $3,7,11,15,19$. The main result of this paper is that Conjecture \ref{con:concfinhitstr} holds for $k=3$ (Corollary \ref{cor:NV3}). New tools have been developed to show this. First we investigate singular DP-reduction \cite{KullmannZhao2012ConfluenceC,KullmannZhao2012ConfluenceJ}, and especially its inversion called ``singular extensions'', in Sections \ref{sec:singvar}, \ref{sec:uhitnsingvar}. The main novel concept of this \Schrift{} is \emph{irreducibility}, an important and intuitive concept, introduced and developed in Section \ref{sec:redsclscl}: one can not factor out a sub-clause-set logically equivalent to a single clause. We extracted it from our work, and later realised that up to the setting it is basically the same as investigated in \cite{Korec1984Covers,BergerFelzenbaumFraenkel1990Covers}. For this \Schrift{} the main point is that FC can be reduced to the irreducible case via induction. This induction still leaves some leeway, and allowing ``nearly-full-subsumption resolution'' in Section \ref{sec:Subsumptionflips} we can handle deficiency 3.

\section{Preliminaries}
\label{sec:Preliminaries}

Most notations and concepts in this section are standard (see the Handbook chapter \cite{Kullmann2007HandbuchMU}), but we provide all definitions, boldfacing those where confusions are possible. We use standard set-theoretical notations and concepts. For example for a set $X$ of sets by $\bc X$ the union of the elements of $X$ is denoted, and by $\bca X$ for $X \ne \es$ their intersection. The symmetric difference of sets $X, Y$ is $X \symdif Y := (X \sm Y) \cup (Y \sm X)$. We use $\NN = \set{x \in \ZZ : x \ge 1}$ and $\NNZ = \NN \cup \set{0}$.

In this \Schrift{} w.l.o.g.\ we use $\Va := \NN$ for the set of variables, that is, variables are just natural numbers, and $\Lit := \ZZ \sm \set{0}$, that is, literals are just non-zero integers, while complementation (logical negation of literals) is just (arithmetical) negation, that is, for $x \in \Lit$ we use $\ol{x} := -x \in \Lit$. For a set $L \sse \Lit$ of literals we use $\ol{L} := \set{\ol{x} : x \in L}$ for elementwise complementation. A \textbf{clause} is a finite set $C \subset \Lit$ of literals, which is ``clash-free'', that is, $C \cap \ol{C} = \es$; the set of all clauses is denoted by $\Cl$. A \textbf{clause-set} is a finite set of clauses, the set of all clause-sets is denoted by $\Cls$. The underlying variable of a literal, given by $\var: \Lit \ra \NN$, is defined as $\var(x) := \abs{x}$ for $x \in \Lit$, while for a clause $C$ let $\var(C) := \set{\var(x) : x \in C} \subset \Va$, and for a clause-set $F$ let $\var(F) := \bc_{C \in F} \var(C) \subset \Va$. For a set $L \sse \Lit$ of literals let $\lit(L) := L \cup \ol{L}$ be the closure under complementation, while for $F \in \Cls$ let $\lit(F) := \lit(\var(F))$. We note here that the actually occurring literals of $F$ are just the elements of $\bc F$. As measures for clause-sets $F$ we use $n(F) := \abs{\var(F)} \in \NNZ$ for the number of variables, and $c(F) := \abs{F} \in \NNZ$ for the number of clauses. The \textbf{deficiency} $\delta(F) \in \ZZ$ is defined as $\delta(F) := c(F) - n(F)$. For $\mc{C} \sse \Cls$ we use notations like $\mc{C}_{\delta=k} := \set{F \in \mc{C} : \delta(F) = k}$. For $F \in \Cls$ and $x \in \Lit$ let $\bmm{F_x} := \set{C \in F : x \in C} \in \Cls$ be the sub-clause-set consisting of all clauses containing literal $x$, and let $\bmm{\ldeg_F(x)} := c(F_x) \in \NNZ$ be the \textbf{literal-degree} of $x$ in $F$, while for $v \in \Va$ the \textbf{variable-degree} is $\bmm{\vdeg_F(v)} := \ldeg_F(v) + \ldeg_F(\ol{v}) \in \NNZ$. A \textbf{full clause} of $F \in \Cls$ is some $C \in F$ with $\var(C) = \var(F)$, while the set of all full clauses over some finite $V \subset \Va$ is denoted by $\bmm{A(V)} := \set{C \in \Cl : \var(C) = V} \in \Cls$. So the set of full clauses of $F \in \Cls$ is $F \cap A(\var(F))$. Furthermore we use $A_n := A(\tb 1n)$ for $n \in \NNZ$. So $A_0 = \set{\bot}$ and $A_1 = \set{\set{1},\set{-1}}$. A \textbf{full variable} of $F \in \Cls$ is some $v \in \var(F)$ such that for all $C \in F$ holds $v \in \var(C)$. So the subsets of $A(V)$ are precisely the clause-sets where every variable is full.

$\Sat$ is the set of satisfiable clause-sets, which are those $F \in \Cls$ such that there is $C \in \Cl$ with $\fa\, D \in F : C \cap D \ne \es$, while $\Usat := \Cls \sm \Sat$ is the set of unsatisfiable clause-sets. So $F \in \Cls$ is unsatisfiable iff for all $C \in \Cl$ there is $D \in F$ with $C \cap D = \es$. Furthermore $\Musat \subset \Usat$, the set of \textbf{minimally unsatisfiable clause-sets}, is the set of all $F \in \Usat$ such that for all $C \in F$ holds $F \sm \set{C} \in \Sat$. In the \Schrift{} we do not use the usual ``partial assignments'', but just use clauses, whose elements in such a context are thought to be set to true. So in the above definition of $\Sat$ the clause $C$ corresponds to a ``satisfying (partial) assignment''. This usage of clauses depends on clauses $C$ not being tautological, i.e., $C \cap \ol{C} = \es$ --- otherwise we had an inconsistency.

$F \in \Cls$ is called \textbf{irredundant}, if for all $C \in F$ there exists a super-clause $D \in \Cl$, $C \sse D$, such that for all $E \in F \sm \set{C}$ holds $D \cap \ol{E} \ne \es$; the set of all irredundant clause-sets is denoted by $\bmm{\Irred} \subset \Cls$. We note that for $F \in \Irred$ and $F' \sse F$ also $F' \in \Irred$ holds. We have $\Musat \subset \Irred$, and indeed $\Musat = \Usat \cap \Irred$. Two clause-sets $F, G$ are \textbf{logically equivalent} iff $\fa\, C \in \Cl : (\fa\, D \in F : C \cap D \ne \es) \Lra (\fa\, D \in G : C \cap D \ne \es)$. So $F \in \Cls$ is irredundant iff there is no $C \in F$ such that $F$ is logically equivalent to $F \sm \set{C}$ iff subsets $F', F'' \sse F$ are logically equivalent only if they are equal.

Two clause-sets $F, G$ are \textbf{isomorphic}, written \bmm{F \cong G}, if there is a bijection (an ``isomorphism'') $f: \lit(F) \ra \lit(G)$ with $f(\ol{x}) = \ol{f(x)}$ for $x \in \lit(F)$ and $G = \set{\set{f(x) : x \in C} : C \in F}$ (see ``mixed symmetries'' in \cite[Section 10.4]{Sak09HBSAT}).

$\Clash$ is the set of \textbf{hitting clause-sets}, i.e., those $F \in \Cls$ such that for all $C, D \in F$, $C \ne D$, holds $C \cap \ol{D} \ne \es$. We have $\Clash \subset \Irred$. The central class for this \Schrift{} is $\bmm{\Uclash} := \Clash \cap \Usat$. Obviously $A(V) \in \Uclash$. If $F \in \Uclash$ has at least two unit-clauses, then $F \cong A_1$. If for $F \in \Cls$ holds $\sum_{C \in F} 2^{-\abs{C}} < 1$, then $F \in \Sat$, while for all $F \in \Clash$ holds $\sum_{C \in F} 2^{-\abs{C}} \le 1$, and for $F \in \Clash \cup \Usat$ holds $F \in \Uclash \Lra \sum_{C \in F} 2^{-\abs{C}} = 1$.

The default interpretation of clause-sets $F$ is as a CNF (conjunction of disjunction), and so the logical conjunction for $F, G \in \Cls$ is just realised by $F \cup G$, while the logical disjunction is union clause-wise:
\begin{definition}\label{def:cor}
  For clause-sets $F, G \in \Cls$ we construct $\bmm{F \cor G} \in \Cls$, the \textbf{combinatorial disjunction} of $F, G$, as the set of all clauses $C \cup D$ for $C \in F$ and $D \in G$ (since clauses are clash-free, only non-clashing pairs $C,D$ are considered here): $F \cor G := \set{C \cup D \mb C \in F \und D \in G \und C \cap \ol{D} = \es}$.
\end{definition}
$F \cor G$ is logically equivalent to the disjunction of $F$ and $G$. So for $G \in \Usat$ we have that $F \cor G$ is logically equivalent to $F$. And $F \cor G \in \Usat \Lra \set{F,G} \sse \Usat$. For a finite $V \subset \Va$ we have $A(V) = \cor_{v \in V} \set{\set{v},\set{\ol{v}}}$. As $\Usat$ is stable under $\cor$, so is $\Clash$, and thus also $\Uclash$.

The \textbf{resolution operation} $\bmm{C \res D} \in \Cl$ for clauses $C, D \in \Cl$ is only partially defined, namely for $\abs{C \cap \ol{D}} = 1$, in which case $C \res D := (C \cup D) \sm \lit(C \cap \ol{D})$, or, in other words, if there is a literal $x$ with $x \in C$, $\ol{x} \in D$, and $(C \cup D) \sm \set{x,\ol{x}}$ is a clause. \textbf{DP-reduction} is denoted for $F \in \Cls$ and $v \in \Va$ by $F \leadsto \bmm{\dpi{v}(F)} := \set{C \res D : C, D \in F, C \cap \ol{D} = \set{v}} \cup \set{C \in F : v \notin \var(C)}$ (also called ``variable elimination''), that is, replacing all clauses containing $v$ by their resolvents. $\Uclash$ behaves well for (general) DP-reductions (\cite{KullmannZhao2012ConfluenceJ}): it is stable, and a sequence of DP-reductions does not depend on the order.

A special case of resolution, where both parent clauses are identical up to the resolution literals, is called ``full subsumption resolution'', and the corresponding resolutions and ``inverse resolutions'' are performed abundantly. Basic theory and applications one finds in \cite[Section 6]{KullmannZhao2010Extremal} and \cite[Section 5]{KullmannZhao2015FullClauses}:
\begin{definition}\label{def:fullsubres}
  Using a slight abuse of language, a \textbf{full subsumption pair} (short ``fs-pair'') is a set $\set{C,D}$ such that $C, D \in \Cl$, $\abs{C \cap \ol{D}} = 1$, and $\abs{C \symdif D} = 2$. A \textbf{full subsumption resolution} (``fs-resolution'') can be performed for $F \in \Cls$, if there is an fs-pair $\set{C,D} \sse F$, such that $C \res D \notin F$, in which case $F$ is called \textbf{full subsumption resolvable} (``fs-resolvable''), and performing the fs-resolution means the transition $F \leadsto (F \sm \set{C,D}) \cup \set{C \res D}$. An fs-resolution is called \textbf{strict}, if no variable is lost in the transition, otherwise \textbf{non-strict}, while if we just speak of ``fs-resolution'', then it may be strict or non-strict. In the other direction we speak of \textbf{(strict/non-strict) full subsumption extension} (``fs-extension''), that is, the transition $F \in \Cls \leadsto F' \in \Cls$, such that $F'$ is (strict/non-strict) fs-resolvable, and the fs-resolution yields $F$.
\end{definition}
In other words, for a clause $C \in F$ and a variable $v \in \Va \sm \var(C)$ we can perform an fs-extension on $C$, replacing $C$ by $C \cup \set{v}, C \cup \set{\ol{v}}$, iff none of these two clauses is already in $F$ (which is guaranteed for irredundant $F$); strictness means $v \in \var(F)$, non-strictness means $v \notin \var(F)$ (i.e., the fs-extension introduces a new variable). Obviously an fs-pair $\set{C,D}$ is logically equivalent to $\set{C \res D}$, and indeed for clauses $C, D \in \Cl$ there exists a clause $E \in \Cl$ such that $\set{C,D}$ is logically equivalent to $\set{E}$ iff either $C \sse D$ or $D \sse C$ or $\set{C,D}$ is an fs-pair. This topic will be taken up again by the notion of a ``clause-factor'' (Section \ref{sec:redsclscl})

\section{Singular variables}
\label{sec:singvar}

\cite[Section 3]{KullmannZhao2012ConfluenceJ} started a systematic investigation into \textbf{singular DP-reduction} (which of course played already an important role in earlier work on MU, e.g.\ \cite{KleineBuening2000SubclassesMU}). A \textbf{singular variable} of $F \in \Cls$ is a variable $v$ with $\min(\ldeg_F(v), \ldeg_F(\ol{v})) = 1$, while a clause-set $F \in \Cls$ is called \textbf{nonsingular} if $F$ does not have singular variables; denoting the set of singular variables of $F$ with $\bmm{\varsing(F)} \sse \var(F)$, thus $F$ is nonsingular iff $\varsing(F) = \es$. The subsets of nonsingular elements of $\Musat$ and $\Uclash$ are denoted by \bmm{\Musatns} and \bmm{\Uclashns}. For $F \in \Cls$ a singular DP-reduction is the transition $F \leadsto \dpi{v}(F)$ for a singular variable $v \in \varsing(F)$. More precisely we call a variable $v$ \textbf{$m$-singular} for $F$ and $m \in \NN$ if $v$ is singular and $\vdeg_F(v) = m+1$; the set of all 1-singular variables of $F$ is denoted by $\varosing(F) \sse \varsing(F)$, while the set of \textbf{non-1-singular variables} is $\bmm{\varnosing(F)} := \varsing(F) \sm \varosing(F)$. By \cite[Lemma 12, Part 2(b)]{KullmannZhao2012ConfluenceJ} we have:
\begin{lemma}[\cite{KullmannZhao2012ConfluenceJ}]\label{lem:characsingDPuhit}
 $\set{\bca F_v, \bca F_{\ol{v}}}$ is an fs-pair for all $F \in \Uclash$, $v \in \varsing(F)$.
\end{lemma}
I.e., let $C \in F$ be the \textbf{main clause} and $D_1, \dots, D_m \in F$ be the \textbf{side clauses} of the $m$-singular variable $v \in \varsing(F)$: Lemma \ref{lem:characsingDPuhit} says $\set{C, \bca_{i=1}^m D_i}$ is an fs-pair. So a 1-singular variable for UHIT is the situation of a non-strict fs-resolution:
\begin{corollary}\label{cor:charac1sUHIT}
  For $F \in \Uclash$ and $v \in \varosing(F)$: $F_v \cup F_{\ol{v}}$ is an fs-pair.
\end{corollary}

\begin{corollary}\label{cor:2suhit}
  Consider $F \in \Uclash$ and a 2-singular variable $v$. Then the side-clauses $D_1, D_2 \in F$ yield an fs-pair $\set{D_1,D_2}$.
\end{corollary}
\begin{proof}
Consider the main clause $C$, w.l.o.g.\ assume $v \in C$, and let $C_0 := C \sm \set{v}$. Then $C_0 \cup \set{\ol{v}} = D_1 \cap D_2$. Since $D_1, D_2$ clash, there is $w \in \var(F)$ with w.l.o.g.\ $w \in D_1$, $\ol{w} \in D_2$, and thus $w \notin \var(C)$. If there would be some other literal, say w.l.o.g.\ $x \in D_2 \sm (C_0 \cup \set{\ol{v},\ol{w}})$, then the assignment setting all literals in $C_0$ to false and setting $v,w,x$ to true would satisfy $F$ (due to $F \in \Clash$). \qed
\end{proof}

\begin{corollary}\label{cor:uhit3fs}
  If $F \in \Uclash$ contains a variable occurring at most three times, then this variable is a singular variable, and $F$ contains an fs-pair.
\end{corollary}
In Lemma \ref{lem:22fns} we give further sufficient criterion for the presence of fs-pairs.

\begin{corollary}\label{cor:Smuonly11sgen}
  For $x, y \in C \in F \in \Uclash$, $x \ne y$: $\ldeg_F(x) = 1 \Ra \ldeg_F(y) \ge 2$.
\end{corollary}
\begin{proof}
Consider $D \in F$ with $\ol{x} \in D$; by Lemma \ref{lem:characsingDPuhit} $y \in D$, thus $\ldeg_F(y) \ge 2$. \qed
\end{proof}

By \cite[Theorem 23]{KullmannZhao2012ConfluenceJ} we know that singular DP-reduction is confluent for $\Uclash$. So we have the retraction $\bmm{\sNF}: \Uclash \ra \Uclashns$, which maps $F \in \Uclash$ to the unique nonsingular $\sNF(F)$ obtainable from $F$ by iterated singular DP-reduction. $\Uclash$ is partitioned into the \textbf{singular fibres} $\sNF^{-1}(F)$ for $F \in \Uclashns$. More generally, by \cite[Theorem 63]{KullmannZhao2012ConfluenceJ} the singularity index $\bmm{\singind(F)} \in \NNZ$ is defined for $F \in \Musat$ as the unique number of singular DP-reductions needed to reduce $F$ to an element of $\Musatns$; for $F \in \Uclash$ the uniqueness of the number of reductions steps also follows with the help of the confluence of sDP-reduction. We have $\singind(F) = c(F) - c(\sNF(F)) = n(F) - n(\sNF(F))$ for $F \in \Uclash$.

Consider $m \in \NN$; a \emph{general $m$-singular extension} of $G \in \Cls$ with $x \in \Lit \sm \lit(F)$ is some $F \in \Cls$ with $\ldeg_F(x) = 1$, $\ldeg_F(\ol{x}) = m$, and $\dpi{\var(x)}(F) = G$. By \cite[Lemma 9]{KullmannZhao2012ConfluenceJ} we know that $F \in \Musat$ implies $G \in \Musat$, and since DP-reduction is satisfiability-equivalent, we have that $G \in \Usat$ implies $F \in \Usat$, however in general $G \in \Musat$ does not imply $F \in \Musat$, since there might be tautological resolvents, and some resolvents might already exist in $F$. This is excluded by the definition of a ``$m$-singular extensions'' in \cite[Definition 5.6]{KullmannZhao2010Extremal}, which we need to generalise in order not just to preserve $\Musat$, but also $\Uclash$. Consider $\mc{C} \sse \Cls$, $G \in \mc{C}$, $m \in \NN$ and $x \in \Lit \sm \lit(G)$. A general $m$-singular extension $F$ of $G$ with $x$ is called an \textbf{$m$-singular $\mc{C}$-extension of $G$ with $x$} if $F \in \mc{C}$ and $c(F) = c(G) + 1$. For ``hitting extensions'' we need to obey Lemma \ref{lem:characsingDPuhit} and obtain:
\begin{lemma}\label{lem:singext}
  For $G \in \Uclash$, $x \in \Lit \sm \lit(F)$ and $m \in \NN$ the \textbf{$m$-singular hitting extensions $F$} (the $m$-singular $\Uclash$-extensions of $G$) are given by choosing some $G' \sse G$ with $c(G') = m$ such that the clause $\bca G'$ clashes with every element of $G \sm G'$, and letting $F := (G \sm G') \cup \set{(\bca G') \cup \set{x}} \cup (\set{\ol{x}} \cor G')$.
\end{lemma}
Two principal choices for $G'$ are always possible (the \emph{trivial singular hitting extensions}): The 1-singular hitting extensions are precisely the non-strict fs-extensions. At the other end, a $c(G)$-singular hitting extension of $G$ adds the unit-clause $\set{x}$ and adds to every other clause the literal $\ol{x}$; these extensions are called \textbf{full singular unit-extensions}. A simple observation:
\begin{lemma}\label{lem:fullsingfs}
  Consider $F \in \Uclash \sm \set{\bot}$ and obtain $F'$ by full singular unit-extension. Then $F'$ has an fs-pair if and only if $F$ has an fs-pair.
\end{lemma}

We conclude this section by some applications to the structure of $\Uclash$, using the minimal var-degree $\minvdeg(F) := \min_{v \in \var(F)} \vdeg_F(v)$, where by \cite{Ku99dKo} for $F \in \Uclashi{\delta=2}$ holds $\minvdeg(F) \in \set{2,3,4}$.
\begin{lemma}\label{lem:specpropd2}
  Consider $F \in \Uclashi{\delta=2}$ with $\minvdeg(F) = 4$.
  \begin{enumerate}
  \item\label{lem:specpropd21} $F$ is singular iff $F$ has a unit-clause iff $F$ is not isomorphic to $\Dt{2}$ or $\Dt{3}$.
  \item\label{lem:specpropd22} $F$ is obtained from $\Dt{2}$ or $\Dt{3}$ by a series of full singular unit-extensions.
  \item\label{lem:specpropd23} $F$ is not fs-resolvable iff $F$ is obtained from $\Dt{3}$ by a series of full singular unit-extensions.
  \end{enumerate}
\end{lemma}
\begin{proof}
\cite[Lemma 5.13]{KullmannZhao2010Extremal} proves Part \ref{lem:specpropd21}. Part \ref{lem:specpropd22} follows by induction, using Part \ref{lem:specpropd21} and the fact, that singular DP-reduction does not decrease the minimum var-degree (\cite[Lemma 5.4]{KullmannZhao2010Extremal}). Finally Part \ref{lem:specpropd23} follows with Lemma \ref{lem:fullsingfs}. \qed
\end{proof}

\begin{corollary}\label{cor:fsd2}
  $F \in \Uclashi{\delta=2}$ is not fs-resolvable iff $F$ is obtained from a clause-set isomorphic to $\Dt{3}$ by a series of full singular unit-extensions (or, equivalently, unit-clause propagation on $F$ yields a clause-set isomorphic to $\Dt{3}$).
\end{corollary}
\begin{proof}
We have $\minvdeg(F) \in \set{2,3,4}$. If $\minvdeg(F) \le 3$, then by Corollary \ref{cor:uhit3fs} $F$ is fs-resolvable, while every clause-set obtained from $\Dt{3}$ by a series of full singular unit-extensions has $\minvdeg(F) \ge 4$. \qed
\end{proof}

Using that all $F \in \Uclashi{\delta=1}$ are fs-resolvable except of $F = \set{\bot}$, we get:
\begin{corollary}\label{cor:fsr5}
  $F \in \Uclash$ with $c(F) \le 5$ is not fs-resolvable iff $F = \set{\bot}$ or $F \cong \Dt{3}$.
\end{corollary}

\section{Number of singular variables vs the singularity index}
\label{sec:uhitnsingvar}

\begin{definition}\label{def:nsv}
  For $F \in \Cls$ let $\bmm{\nsv(F)} := \abs{\varsing(F)} \in \NNZ$, while $\bmm{\nosv(F)} := \abs{\varosing(F)} \in \NNZ$ and $\bmm{\nnosv(F)} := \abs{\varnosing(F)} \in \NNZ$.
\end{definition}
Thus $\nsv(F) = \nosv(F) + \nnosv(F)$. We show that for $F \in \Uclash$ with ``large'' $\singind(F)$ also $\nsv(F)$ must be ``large'' (proving \cite[Conjecture 76]{KullmannZhao2012ConfluenceJ}). First an auxiliary lemma, showing how we can reduce the number of singular variables together with the singularity index:
\begin{lemma}\label{lem:controlsdpuhit1}
  Consider $F \in \Uclash$ with $\varsing(F) \not= \es$. Then there is a singular tuple $\vec{v} = (v_1,\dots,v_m)$ for $F$ with $1 \le m \le 2$ such that $\varsing(\dpi{\vec{v}}(F)) \sse \varsing(F) \sm \var(\set{v_1,\dots,v_m})$ (recall the order-independency of DP for $\Uclash$). More specifically, we can choose $\vec{v} = (v)$ for every $v \in \varnosing(F)$; assume $\varnosing(F) = \es$ in the sequel. For $v \in \varosing(F)$ there is a clause $C$ such that $C \cup \set{v}, C \cup \set{\ol{v}} \in F$ (Corollary \ref{lem:characsingDPuhit}). We can choose again $\vec{v} = (v)$ if for all $x \in C$ we have $\ldeg_F(x) \ge 3$. Otherwise consider some $x \in C$ with $\ldeg_F(x) = 2$ and $\ldeg_F(\ol{x}) \ge 2$. Now we can choose $\vec{v} = (v, \var(x))$.
\end{lemma}

\begin{lemma}\label{lem:nsvsi}
  For $F \in \Uclash$ holds $\nsv(F) \ge \frac 12 \singind(F)$.
\end{lemma}
\begin{proof}
  We use induction on $\singind(F)$. The statement holds trivially for $\singind(F) = 0$, and so assume $\singind(F) > 0$. Consider a singular tuple $\vec{v} = (v_1,\dots,v_m)$ for $F$ according to Lemma \ref{lem:controlsdpuhit1}, and let $F' := \dpi{\vec{v}}(F)$ (note that $\singind(F') = \singind(F) - m$). Applying the induction hypothesis to $F'$ we get $\nsv(F') \ge \frac 12 \cdot \singind(F') = \frac 12 \cdot (\singind(F) - m) \ge \frac 12  \cdot \singind(F) - 1$, and thus $\nsv(F) \ge \nsv(F') + 1 \ge \frac 12 \cdot \singind(F)$. \qed
\end{proof}

So we get $\singind(F) \le 2 \nsv(F)$ for $F \in \Uclash$. This can be refined:
\begin{corollary}\label{cor:nsvsi}
  For $F \in \Uclash$ holds $\singind(F) \le 2 \nosv(F) + \nnosv(F)$.
\end{corollary}
\begin{proof}
  We perform first sDP-reduction (only) on the non-1-singular variables, until they all disappear, obtaining $F' \in \Uclash$. By \cite[Corollary 25, Part 1]{KullmannZhao2012ConfluenceJ}, we have $\var(F) \sm \var(F') \sse \varsing(F)$ and $\varsing(F') \sse \varsing(F)$. We now apply Lemma \ref{lem:nsvsi} to $F'$. \qed
\end{proof}

As an application we obtain that after an fs-resolution on a nonsingular UHIT, three singular DP-reductions are sufficient to remove all singularities:
\begin{lemma}\label{lem:si2sr}
  Consider an fs-resolvable $F \in \Uclashns$, where fs-resolution yields $F'$ (thus $F' \in \Uclash$). Then $\singind(F') \le 3$.
\end{lemma}
\begin{proof}
Let $F' = (F \sm \set{C,D}) \cup \set{R}$ with $R := (C \cup D) \sm \set{v,\ol{v}}$. Assume $\singind(F') \ge 4$. Thus by Lemma \ref{lem:nsvsi} we have $\nsv(F') \ge 2$. By Corollary \ref{cor:Smuonly11sgen} follows $\varsing(F') = \set{v,w}$, where $w \in \var(R)$, since only at most literal of $R$ can have become singular in $F'$. But since $F$ is nonsingular, the variable $w$ is non-1-singular, contradicting Corollary \ref{cor:nsvsi}. \qed
\end{proof}

\section{Reducing sub-clause-sets to clauses: ``factors''}
\label{sec:redsclscl}

What clause-sets $F$ are logically equivalent to clauses $C$ ? If in some $F \in \Cls$ we find some $F' \sse F$ (logically) equivalent to $C$, then $F$ is equivalent to $(F \sm F') \cup \set{C}$. In preparation for the easy answer, note that for all $F \in \Cls \sm \set{\top}$ holds $F = \set{\bca F} \cor \set{D \sm \bca F : D \in F}$.
\begin{lemma}\label{lem:clseqcl}
  For $F \in \Cls$ and $C \in \Cl$ the following properties are equivalent:
  \begin{enumerate}
  \item\label{lem:clseqcl1} $F$ is logically equivalent to $\set{C}$.
  \item\label{lem:clseqcl2} $F \ne \top$, $\bca F = C$, and $\set{D \sm C : D \in F} \in \Usat$.
  \item\label{lem:clseqcl3} There is $G \in \Usat$, $\var(G) \cap \var(C) = \es$, such that $F = \set{C} \cor G$.
  \item\label{lem:clseqcl4} There is $G \in \Usat$ with $F = \set{C} \cor G$.
  \end{enumerate}
\end{lemma}

Clause-sets equivalent to clauses we call ``clause-factors'':
\begin{definition}\label{def:clfactor}
  A \textbf{clause-factor} is some $F \in \Cls \sm \set{\top}$ with $\set{C \sm \bca F : C \in F} \in \Usat$. The \textbf{clause-factors of $F \in \Cls$} are the sub-clause-sets of $F$ which are themselves clause-factors. A clause-factor $F$ of $F'$ is \textbf{trivial} if $c(F) = 1$ or $F \in \Usat \wedge F' = F$, otherwise \textbf{nontrivial}. The \textbf{intersection of a clause-factor} $F$ is $\bca F \in \Cl$. The \textbf{residue of a clause-factor} $F$ is $\set{C \sm \bca F : C \in F} \in \Usat$; a \textbf{residual clause-factor of $F \in \Cls$} is the residue of a clause-factor of $F$.
\end{definition}
Subsets of irredundant clause-sets are irredundant again, and thus clause-factors of irredundant clause-sets are irredundant (as clause-sets):
\begin{lemma}\label{lem:clfacinh}
  Consider a residual clause-factor $G$ of $F \in \Cls$. If $F$ is irredundant, then $G \in \Musat$. If $F \in \Clash$, then $G \in \Uclash$.
\end{lemma}

\subsection{Clause-factorisations}
\label{sec:factorisation}

We see that a combinatorial disjunction $F \cor G$ is the union of the clause-factors $\set{C} \cor G$ for $C \in F$. If we want just to single out a single clause of $F$ for this operation, keeping the rest of $F$, we do this by ``pointing'' $F$:
\begin{definition}\label{def:pointedcor}
  A \textbf{pointed clause-set} is a pair $(F,C) \in \Cls \times \Cl$ with $C \in F$. For a pointed clause-set $(F,C)$ and $G \in \Cls$ we define the \textbf{pointed combinatorial disjunction} (``pcd''; recall Definition \ref{def:cor}) as
  \begin{displaymath}
    \bmm{(F,C) \cor G} := (F \sm \set{C}) \cup (\set{C} \cor G) \in \Cls.
  \end{displaymath}
\end{definition}
The simplest choice for $F$ is $\set{C} \cor G = (\set{C},C) \cor G$. The two simplest choices for $G$ are $(F,C) \cor \top = F \sm \set{C}$ and $(F,C) \cor \set{\bot} = F$. Using the interpretation of clause-sets as CNFs, $(F,C) \cor G$ is logically equivalent to $(F \sm \set{C}) \wedge (C \vee G)$; so if $G$ is unsatisfiable, then $(F,C) \cor G$ is logically equivalent to $F$.

\begin{definition}\label{def:clausefac}
  A pointed combinatorial disjunction $(F,C) \cor G$ (according to Definition \ref{def:pointedcor}) is called a \textbf{clause-factorisation (of $F$ and $G$ via $C$)}, if $\var(C) \cap \var(G) = \es$, the union is disjoint (i.e., $(F \sm \set{C}) \cap (\set{C} \cor G) = \es$), and furthermore $G \in \Usat$ holds. In a clause-factorisation $(F,C) \cor G$ we call $\set{C} \cor G$ \textbf{the factor}, $G$ the \textbf{residual factor}, and $F$ the \textbf{cofactor}. A clause-factorisation is \textbf{trivial}, if $\set{F,G} \cap \set{\set{\bot}} \ne \es$, otherwise \textbf{nontrivial}.
\end{definition}
The trivial clause-factorisations are $(F,C) \cor \set{\bot} = F$ and $(\set{\bot},\bot) \cor G = G$ for $F \in \Cls$ and $G \in \Usat$. Correspondingly, for the trivial factor $\set{C}$ of $F \in \Cls$, $C \in F$, the intersection is $C$, the residue is $\set{\bot}$, and the cofactor is $F$, while for the trivial factor $G$ of $G \in \Usat$ the intersection is $\bot$ and the cofactor is $\set{\bot}$. Directly from the definitions we obtain the basic properties:

\begin{lemma}\label{lem:propclfac}
  Consider $F \in \Cls$ and a clause-factorisation $F = (F_0,C) \cor G$.
  \begin{enumerate}
  \item\label{lem:propclfac2} $\delta(F) = \delta(F_0) + \delta(G) - 1 + \abs{\var(F_0) \cap \var(G)}$.
  \item\label{lem:propclfac1} $\set{F_0,G} \subset \Musat \Lra F \in \Musat$.
  \item\label{lem:propclfac3} If $F \in \Musat$, then:
    \begin{enumerate}
    \item\label{lem:propclfac3a} $1 \le \delta(F_0) \le \delta(F)$ and $1 \le \delta(G) \le \delta(F)$.
    \item\label{lem:propclfac3b} $\delta(F) = \delta(F_0)$ iff $\var(F_0) \cap \var(G) = \es$ and $G \in \Musati{\delta=1}$.
    \item\label{lem:propclfac3c} $\delta(F) = \delta(G)$ iff $\var(F_0) \cap \var(G) = \es$ and $F_0 \in \Musati{\delta=1}$.
    \item\label{lem:propclfac3d} If $F$ is nonsingular and $F \ne \set{\bot}$:
      \begin{enumerate}
      \item\label{lem:propclfac3di} If $\varosing(F_0) \cap \var(G) = \es$, then $\varosing(F_0) = \es$.
      \item\label{lem:propclfac3dii} If $\var(F_0) \cap \varsing(G) = \es$, then $G$ is nonsingular.
      \item\label{lem:propclfac3diii} If the factorisation in nontrivial: $\delta(F_0) < \delta(F)$ and $\delta(G) < \delta(F)$.
      \end{enumerate}
    \end{enumerate}
  \end{enumerate}
\end{lemma}

The relation between clause-factorisations and -factors is now easy to see:
\begin{lemma}\label{lem:ntrclsub}
  $F \in \Cls$ allows a nontrivial clause-factorisation iff $F$ contains a nontrivial clause-factor.
\end{lemma}

Following \cite{Korec1984Covers,BergerFelzenbaumFraenkel1990Covers} (introducing ``irreducibility'' for covers of the integers resp.\ cell partitions of lattice parallelotops), we introduce the fundamental notion of ``clause-irreducible clause-sets'', not allowing non-trivial clause-factorisations:
\begin{definition}\label{def:clprime}
  A clause-set $F \in \Cls$ is called \textbf{clause-irreducible}, if every clause-factor is trivial, otherwise $F$ is called \textbf{clause-reducible}; the set of all clause-irreducible clause-sets is denoted by $\bmm{\Clpr} \subset \Cls$.
\end{definition}

\subsection{Clause-factors for UHIT}
\label{sec:primeuhit}

\begin{definition}\label{def:clprimeuhit}
  In this \Schrift{} we are especially concerned with $\Uclash$, and we call the subset given by the clause-irreducible elements $\bmm{\Pruclash} := \Clpr \cap \Uclash$.
\end{definition}
So $\Pruclashi{n \le 1} = \Uclashi{n \le 1} = \set{\set{\bot}} \cup \set{\set{\set{v},\set{\ol{v}}} : v \in \Va}$.

\begin{example}\label{exp:nonprimeF2}
  $\Dt{2} = \set{\set{1,2},\set{-1,-2},\set{-1,2},\set{-2,1}}$ is clause-reducible, and the nontrivial clause-factors are the $\binom 42 - 2 = 4$ 2-element subsets of $F$ where the two clauses have precisely one clash (these are the fs-pairs). So $\Pruclashi{n = 2} = \es$.
\end{example}

\begin{lemma}\label{lem:equivcharacfacuh}
  Consider $F \in \Uclash$ and a non-empty subset $\top \ne F' \sse F$, and let $F'' := (F \sm F') \cup \set{\bca F'}$ be the ``cofactor''. We note that this union is disjoint, since $F \in \Clash$. The following conditions are equivalent:
  \begin{enumerate}
  \item\label{lem:equivcharacfacuh1} $F'$ is a clause-factor of $F$.
  \item\label{lem:equivcharacfacuh2} The intersection of $F'$ clashes with every other clause, i.e., $F'' \in \Clash$.
  \item\label{lem:equivcharacfacuh3} $F'' \in \Uclash$.
  \end{enumerate}
\end{lemma}
\begin{proof}
Part \ref{lem:equivcharacfacuh1} implies Part \ref{lem:equivcharacfacuh3}: If $F'$ is a factor of $F$, then $F''$ is unsatisfiable, since $\bca F'$ subsumes all clauses of $F'$, and $F''$ is a hitting clause-set, since if there would be some $C \in F \sm F'$ without a clash with $\bca F'$, then setting all literals in $C$ to false would be a satisfying assignment for the residue. Trivially Part \ref{lem:equivcharacfacuh3} implies \ref{lem:equivcharacfacuh2}. Finally assume $F'' \in \Clash$, but that the residue $\set{C \sm \bca F' : C \in F'} \in \Sat$. So then there is a clause $D$ with $\var(D) \cap \var(\bca F') = \es$, which has a clash with every clause in $F'$, and so $D \cup \bca F'$ is a clause with a clash with every clause of $F$, contradicting unsatisfiability of $F$. \qed
\end{proof}

Factors of UHITs are basically the same as singular extensions resulting in unsatisfiable hitting clause-sets (up to the choice of the extension-variable):
\begin{lemma}\label{lem:corrfactoruhitext}
  Consider $F \in \Uclash$. Then up to the choice of the extension variable, the singular hitting extensions of $F$ are given according to Lemma \ref{lem:singext} by some nonempty $G \sse F$, and by Lemma \ref{lem:equivcharacfacuh} these subsets are precisely the factors $F'$ of $F$. So the singular $m$-hitting-extensions for $m \ge 1$ correspond 1-1 to the factors $F'$ of $F$ with $c(F') = m$. Especially, the trivial factors of $F$ correspond 1-1 to the trivial singular hitting extensions of $F$, namely the factors of size $1$ correspond to the 1-extensions, and the factors of size $c(F)$ correspond to the full singular unit-extensions.
\end{lemma}

\begin{corollary}\label{cor:characprimeext}
  A clause-set $F \in \Uclash$ is irreducible if and only if every singular hitting-extension is trivial.
\end{corollary}

Singular variables or full variables yield factors as follows:
\begin{lemma}\label{lem:singfact}
  Consider $F \in \Uclash$ and $v \in \var(F)$. If $v$ is a singular variable or a full variable of $F$, then $F_v$, $F_{\ol{v}}$ and $F_v \cup F_{\ol{v}}$ are factors of $F$.
\end{lemma}
\begin{proof}
First consider that $v$ is a singular variable of $F$, and assume w.l.o.g.\ that $\ldeg_F(v) = 1$. Then trivially $F_v$ is a factor of $F$. Let $C$ be the main clause of $v$, where w.l.o.g.\ $v \in C$, and let $D := C \sm \set{v}$. Now $\bca F_{\ol{v}} = D \cup \set{\ol{v}}$ (since $F$ is hitting), and thus $F_{\ol{v}}$ is a factor of $F$ (since $C$ clashes with every other clause). Finally $\bca (F_v \cup F_{\ol{v}}) = D$, and thus also $F_v \cup F_{\ol{v}}$ is a factor. Now assume that $v$ is a full variable of $F$. Then $F_v \cup F_{\ol{v}} = F$, while $\bca F_v = \set{v}$ (and $\bca F_{\ol{v}} = \set{\ol{v}}$; otherwise $F$ would be satisfiable), and thus also $F_v, F_{\ol{v}}$ are factors. \qed
\end{proof}

There are two other classes of easily recognisable factors:
\begin{lemma}\label{lem:2subpfac}
  Consider $F \in \Uclash$. The factors $F'$ with $c(F') = 2$ are precisely the fs-pairs (recall Definition \ref{def:fullsubres}) contained in $F$.
\end{lemma}
\begin{proof}
First consider an fs-pair $F' := \set{C \cup \set{v}, C \cup \set{\ol{v}}} \sse F$; note that $\bca F' = C$. If there would be $D \in F \sm F'$ with $C \cap \ol{D} = \es$, then $D$ would also be clash-free with one element of $F'$, contradicting the hitting condition.

Now consider a factor $F'$ with $c(F') = 2$, and let $C := \bca F'$. Then $(F \sm F') \cup \set{C} \in \Uclash$ by Lemma \ref{lem:equivcharacfacuh}. Due to $\sum_{C \in F} 2^{-\abs{C}} = 1$ we have that $\abs{D} = \abs{C} + 1$ for $D \in F$, and because of the hitting condition thus $F'$ must be an fs-pair. \qed
\end{proof}

\begin{lemma}\label{lem:factorcm1}
  Consider $F \in \Uclash$. Then the factors $F'$ with $c(F') = c(F)-1$ are precisely given by $F' = F \sm \set{C}$ for $C \in F$ with $\abs{C}=1$ (unit-clauses).

  So if $c(F)=2$ (i.e., $F = \set{\set{v},\set{\ol{v}}}$ for some $v \in \Va$), then there are precisely two such factors, while otherwise there can be at most one such factor.
\end{lemma}
\begin{proof}
Consider a factor $F'$ of $F$ with $c(F') = c(F)-1$, and let $C := \bca F'$. Since $F' \ne F$, we have $\abs{C} \ge 1$. If $\abs{C} \ge 2$, then $F$ would be satisfiable (note that $C$ clashes with the clause in $F \sm F'$). So there is a literal $x$ with $C = \set{x}$. Now the clause of $F \sm F'$ must be $\set{\ol{x}}$, since otherwise again $F$ would be satisfiable. The remaining assertions follow easily. \qed
\end{proof}

By Lemmas \ref{lem:singfact}, \ref{lem:2subpfac}, \ref{lem:factorcm1}:
\begin{lemma}\label{lem:trivirred}
  A clause-irreducible $F \in \Uclash$ is singular iff it has a full variable iff it has an fs-pair iff it has a unit-clause iff $F \cong A_1$.
\end{lemma}

Since every $F \in \Uclashi{\delta=1}$ with $n(F) > 0$ is fs-resolvable, $\Pruclashi{\delta = 1} = \Uclashi{n \le 1}$. Having no fs-pair resp.\ no full variables has further consequences:
\begin{lemma}\label{lem:smallfactors}
  If $F \in \Uclash$ has no fs-pair, then $F$ has no nontrivial clause-factor $F'$ with $c(F') \le 4$, while if $F$ has no full variable, then $F$ has no nontrivial clause-factor $F'$ with $c(F') \ge c(F) - 3$.
\end{lemma}
\begin{proof}
If there would be a nontrivial clause-factor $F'$ with $c(F') \le 4$, then by Corollary \ref{cor:fsr5} the residue would have an fs-pair, and then also $F$ had one. And if there would be $F'$ with $c(F') \ge c(F) - 3$, then the cofactor would be an element of $\Uclash$ with at most four clauses, and thus had a full variable. \qed
\end{proof}

\begin{lemma}\label{def:pruclash2}
  Up to isomorphism there is one element in $\Pruclashi{\delta = 2}$, namely $\Dt{3}$.
\end{lemma}
\begin{proof}
Up to isomorphism there are precisely two elements in $\Uclashnsi{\delta=2}$, namely $\Dt{2}$, which is clause-reducible (Example \ref{exp:nonprimeF2}), and $\Dt{3}$, which has no fs-pair, and thus by Lemma \ref{lem:smallfactors} is clause-irreducible. \qed
\end{proof}

\subsection{Reducing FC to the irreducible case}
\label{sec:redfinirr}

Strengthening Lemma \ref{lem:propclfac} (again with simple proofs):
\begin{lemma}\label{lem:propclfacuhit}
  Consider $F \in \Cls$ and a clause-factorisation $F = (F_0,C) \cor G$.
  \begin{enumerate}
  \item\label{lem:propclfacuhit1} $\set{F_0,G} \subset \Uclash \Lra F \in \Uclash$.
  \item\label{lem:propclfacuhit3} Assume $F \in \Uclashns$.
    \begin{enumerate}
    \item\label{lem:propclfacuhit2} If $F$ is not strictly fs-resolvable (recall Definition \ref{def:fullsubres}), then $\varosing(F_0) = \varosing(G) = \es$.
    \item\label{lem:decompoUHIT2} $\nsv(F_0) \le \abs{\varsing(F_0) \cap \var(G)} + 1 \le \abs{\var(F_0) \cap \var(G)} + 1 \le \delta(F)$.
    \item\label{lem:decompoUHIT3} $\nsv(G) \le \abs{\var(F_0) \cap \varsing(G)} \le \abs{\var(F_0) \cap \var(G)} \le \delta(F)-1$.
    \end{enumerate}
  \end{enumerate}
\end{lemma}

\begin{theorem}\label{thm:uppboundNV}
  Consider $F \in \Uclashns \sm \Pruclash$ and a nontrivial clause-factorisation $F = (F_0,C) \cor G$.  Let $k := \delta(F)$ (and so $k \ge 2$).
  \begin{enumerate}
  \item\label{thm:uppboundNV1} If $F$ is strictly fs-resolvable, then $n(F) \le \maxnhitdef(k-1) + 3$.
  \item\label{thm:uppboundNV2} Otherwise $n(F) \le \maxnhitdef(\delta(F_0)) + \maxnhitdef(\delta(G)) + \abs{\var(F_0) \cap \var(G)} + 1$.
  \item\label{thm:uppboundNV3} Assume that $\fa\, k' \ge 2 : k' < k \Ra \maxnhitdef(k') = 4 k' - 5$.
    \begin{enumerate}
    \item\label{thm:uppboundNV3a} If $F$ is strictly fs-resolvable, then $n(F) \le 4 k - 6$.
  \item\label{thm:uppboundNV3b} Otherwise $n(F) \le 4 k - 5 - 3 \cdot \abs{\var(F_0) \cap \var(G)} \le 4 k - 5$.
    \end{enumerate}
  \end{enumerate}
\end{theorem}
\begin{proof}
Part \ref{thm:uppboundNV1}: Perform a strict fs-resolution for $F$, obtaining $F'$ (with $\delta(F') = \delta(F) - 1$); by Lemma \ref{lem:si2sr} we get $n(F) = n(F') \le n(\sNF(F')) + 3 \le \maxnhitdef(\delta(F)-1) + 3$. Part \ref{thm:uppboundNV2}: Assume that $F$ is not strictly fs-resolvable. So by Lemma \ref{lem:propclfacuhit}, Part \ref{lem:propclfacuhit2}, we get $\varosing(F_0) = \varosing(G) = \es$. Let $s := \abs{\var(F_0) \cap \var(G)}$. We have
\begin{multline*}
  n(F) = n(F_0) + n(G) - s = (n(\sNF(F_0)) + \singind(F_0)) + (n(\sNF(G)) + \singind(G)) - s \le\\
  (\maxnhitdef(\delta(F_0)) + \singind(F_0)) + (\maxnhitdef(\delta(G)) + \singind(G)) - s.
\end{multline*}
By Corollary \ref{cor:nsvsi} holds $\singind(F_0) \le \nsv(F_0)$ and $\singind(G) \le \nsv(G)$, where by Lemma \ref{lem:propclfacuhit}, Parts \ref{lem:decompoUHIT2}, \ref{lem:decompoUHIT3}, we have $\nsv(F_0) \le s + 1$ and $\nsv(G) \le s$, which completes the proof.

Part \ref{thm:uppboundNV3a} follows from Part \ref{thm:uppboundNV1} for $k \ge 3$: $n(F) \le \maxnhitdef(k-1) + 3 \le 4 (k-1) - 5 + 3 = 4 k - 6$. And for $k=2$ we get $F \cong \Dt{2}$, since $\Dt{3}$ is not fs-resolvable, and thus $2 = n(F) = 4 k - 6$. For Part \ref{thm:uppboundNV3b} we notice that now $k \ge 3$ holds, since $\Dt{3}$ is irreducible by Lemma \ref{def:pruclash2}. Lemma \ref{lem:propclfac}, Part \ref{lem:propclfac2} yields $k = \delta(F_0) + \delta(G) + s - 1$, where by Lemma \ref{lem:propclfac}, Part \ref{lem:propclfac3diii}: $1 \le \delta(F_0) \le k-1$ and $1 \le \delta(G) \le k-1$. By Part \ref{thm:uppboundNV2} we know $n(F) \le \maxnhitdef(\delta(F_0)) + \maxnhitdef(\delta(G)) + s + 1$. And by Lemma \ref{lem:propclfacuhit}, Part \ref{lem:propclfacuhit2}, we get $\varosing(F_0) = \varosing(G) = \es$, and thus $\delta(F_0), \delta(G) \ge 2$. Now $\maxnhitdef(\delta(F_0)) + \maxnhitdef(\delta(G)) + s + 1 = 4 (k-s+1) - 2 \cdot 5 + s + 1 = 4 k - 5 - 3 s$. \qed
\end{proof}

\begin{corollary}\label{cor:uppboundNV}
  Conjecture \ref{con:basicn} is equivalent to the statement, that for all $k \ge 3$ we have $\sup \set{n(F) : F \in \Pruclashi{\delta=k}} < +\infty$. And Conjecture \ref{con:concfinhitstr} is equivalent to the statement, that for all $k \ge 3$ we have $\sup \set{n(F) : F \in \Pruclashi{\delta=k}} \le 4 k - 5$.
\end{corollary}
In Corollary \ref{cor:uppboundNVnfs} we will further restrict the critical cases.

\section{Subsumption-flips}
\label{sec:Subsumptionflips}

Recall that $C, D \in \Cl$ are \emph{full-subsumption resolvable} (``fs-resolvable''; Definition \ref{def:fullsubres}) iff $C, D$ are resolvable and $\abs{C \symdif D} = 2$. In the following we write at places $A \addcup B := A \cup B$ in case $A \cap B = \es$.

\begin{definition}\label{def:nearlyfull}
  Clauses $C, D$ are \textbf{nearly-full-subsumption resolvable} (nfs-resolvable), and $\set{C,D}$ is an \textbf{nfs-pair}, if $C, D$ are resolvable and $\abs{C \symdif D} = 3$.
\end{definition}
$C, D$ are nfs-resolvable iff there is $E \in \Cl$ and $x, y \in \Lit$, $\var(x) \ne \var(y)$, with $\set{C,D} = \set{E \addcup \set{x}, E \addcup \set{\ol{x},y}}$; we call $x$ the \emph{resolution literal}, $y$ the \emph{side literal}, and $E$ the \emph{common part}.

\begin{definition}\label{def:nfsflip}
  For an nfs-pair $\set{C,D} = \set{E \addcup \set{x}, E \addcup \set{\ol{x},y}}$, the \textbf{nfs-flip} is the unordered pair $\set{E \addcup \set{x,\ol{y}}, E \addcup \set{y}}$ (in the clause with the side literal remove the resolution literal, and for the other clause add the complemented side literal). An $F \in \Cls$ is called \textbf{nfs-resolvable}, if there is an nfs-pair $\set{C,D} \sse F$, while none of the two clauses of the nfs-flip is in $F$. For nfs-resolvable $F$ \textbf{on} $\set{C,D} \sse F$, the nfs-flip replaces these two clauses by the result of the nfs-flip (so the number of clauses and the set of variables stays unaltered).
\end{definition}
If $\set{C,D} = \set{E \addcup \set{x}, E \addcup \set{\ol{x},y}}$ are nfs-resolvable, then the result $\set{C',D'}$ of the nfs-flip is again nfs-resolvable, and the nfs-flip yields back $\set{C,D}$. We can simulate the nfs-flip as follows. Performing one strict fs-extension, we obtain $\set{E \addcup \set{x,y}, E \addcup \set{x,\ol{y}}, E \addcup \set{\ol{x},y}}$. Now precisely two strict fs-resolutions are possible, yielding either the original $\set{C,D}$ or the nfs-flip $\set{E \addcup \set{x,\ol{y}}, E \addcup \set{y}}$. So, if $F \in \Uclash$ contains an nfs-pair $\set{C, D} \sse F$ and we replace the pair by its nfs-flip, then we obtain $F' \in \Uclash$, which we say is obtained by \emph{one nfs-flip} from $F$. An nfs-pair $\set{C,D}$ and its flip $\set{C',D'}$ are logically equivalent. Nfs-flips for $F \in \Cls$ leave the measures $n,c,\ell,\delta$ invariant, also the distribution of clause-sizes, while changing precisely the variable degree of two variables of $F$, one goes up and one goes down by one.

\begin{example}\label{exp:nfsflipDT3}
  $\Dt{3} = \set{\set{1,2,3},\set{-1,-2,-3},\set{-1,2},\set{-2,3},\set{-3,1}} \in \Pruclashi{\delta=2}$ is nfs-resolvable, for example the nfs-flip on the first and the third clause in $\Dt{3}$ is $F := \set{\set{2,3},\set{-1,-2,-3},\set{-1,2,-3},\set{-2,3},\set{-3,1}}$. Now $F$ has several nontrivial clause-factors, namely there are two strict fs-pairs, where fs-resolution yields elements of $\Uclashi{\delta=1}$, and $1$ is a $2$-singular variable of $F$ (with $\sNF(F) \cong \Dt{2}$), yielding one further nontrivial clause-factor.
\end{example}

\begin{definition}\label{def:nfsdecompo}
  $F \in \Pruclash$ is \textbf{nfs-reducible}, if via a series of nfs-flips $F$ can be transformed into a clause-reducible clause-set, otherwise $F$ is \textbf{nfs-irreducible}; the set of all nfs-irreducible elements of $\Pruclash$ is denoted by $\bmm{\Ipruclash} \subset \Pruclash$.
\end{definition}

By Lemma \ref{def:pruclash2} and Example \ref{exp:nfsflipDT3}:
\begin{lemma}\label{thm:nfsdecompo2}
  $\Ipruclashi{\delta=2} = \es$.
\end{lemma}

If after an nfs-flip we obtain non-singularity, then it is of the easiest form, and after re-singularisation we have clause-reducibility with additional properties:
\begin{lemma}\label{lem:nfsflspec}
  Consider $F \in \Uclashns$ with an nfs-flip $F'$. Then $\singind(F') \le 1$. Assume that $F$ not fs-resolvable and $\singind(F') = 1$, and let $G := \sNF(F')$. There is a non-trivial clause-factorisation $G = (G_0,C) \cor H$, such that $\var(G_0) \cap \var(H) \ne \es$.
\end{lemma}
\begin{proof}
Consider an nfs-pair $\set{E \cup \set{x}, E \cup \set{\ol{x},y}} \sse F$, and thus $E \cup \set{x,\ol{y}}, E \cup \set{y} \in F'$; we have $\ldeg_{F'}(\ol{x}) = \ldeg_F(\ol{x}) - 1$, $\ldeg_{F'}(\ol{y}) = \ldeg_F(\ol{y}) + 1$, while all other literal degrees remain the same. So the only possibility for a singularity in $F'$ is that $\ldeg_{F'}(\ol{x}) = 1$, and then $\varsing(F') = \varnosing(F') = \set{\var(x)}$, whence in general $\singind(F') \le 1$, proving the first assertion. Now consider the remaining assertions.

Consider the (single) $\ol{x}$-occurrence in $F'$ (which has been transferred unchanged from $F$). By Lemma \ref{lem:characsingDPuhit} and the necessity of a clash with the second $\ol{x}$-occurrence in $F$, this clause is $E' \cup \set{\ol{x},\ol{y}} \in F' \cap F$, where due to $F$ not being fs-resolvable we have $E' \subset E$. Consider the nontrivial factor $F'_x$ of $F'$ according to Lemma \ref{lem:singfact}: We have $E \cup \set{x,\ol{y}} \in F'_x$, while by Lemma \ref{lem:characsingDPuhit} the intersection of $F'_x$ is $E' \cup \set{x,\ol{y}}$. Note that $E \cup \set{y} \in F' \sm F'_x$.

Obtain $F''_x$ from $F'_x$ by removing all occurrences of $x$. Now $G = \dpi{\var(x)}(F')$ is obtained from $F'$ by removing the clause $E' \cup \set{\ol{x},\ol{y}}$, and replacing $F'_x$ by $F''_x$. $F''_x$ is a nontrivial factor of $G$, and via Lemma \ref{lem:corrfactoruhitext} we obtain the sought nontrivial clause-factorisation: $C := E' \cup \set{\ol{y}}$, $H := \set{D \sm C : D \in F''_x}$ and $G_0 := G \sm F''_x$. Due to $E \sm E' \in H$ and $E \cup \set{y} \in G_0$ there is a common variable. \qed
\end{proof}

\begin{theorem}\label{thm:equivconjnfs}
  Consider $k \ge 3$, and assume that $\fa\, k' \in \NN_{\ge 3} : k' < k \Ra \maxnhitdef(k') = 4 k - 5$. Then $\maxnhitdef(k) = 4 k - 5$ is equivalent to the statement, that for all $F \in \Ipruclashi{\delta=k}$ holds $n(F) \le 4 k - 5$.
\end{theorem}
\begin{proof}
Clearly the condition is necessary, and it remains to show that it is sufficient. By Theorem \ref{thm:uppboundNV}, Part \ref{thm:uppboundNV3}, it remains to consider $F_0 \in \Pruclashi{\delta=k}$ and to show $n(F_0) \le 4 k - 5$, and by the condition we can assume that $F_0$ is nfs-reducible. So there exists a series $F_0, \dots, F_m$, $m \ge 1$, such that $F_{i+1}$ is an nfs-flip of $F_i$, and where $F_m$ is clause-reducible, while $F_{m-1}$ is clause-irreducible. If $F_m$ is nonsingular, then we are done (as before), and so assume that $F_m$ is singular. We apply Lemma \ref{lem:nfsflspec}, with $F := F_{m-1}$ and $F' := F_m$, while $G = \sNF(F_m)$ with $n(G) = n(F_0) - 1$. If $G$ is strictly fs-resolvable, then by Theorem \ref{thm:uppboundNV}, Part \ref{thm:uppboundNV3a} we get $n(G) \le 4 k - 6$, so assume that $G$ is not strictly fs-resolvable. Now by Theorem \ref{thm:uppboundNV}, Part \ref{thm:uppboundNV3b} we get $n(G) \le 4 k - 5 - 3 \cdot 1 = 4k - 8$. \qed
\end{proof}

\begin{corollary}\label{cor:uppboundNVnfs}
  Conjecture \ref{con:basicn} is equivalent to the statement, that for all $k \ge 3$ we have $\sup \set{n(F) : F \in \Ipruclashi{\delta=k}} < +\infty$. And Conjecture \ref{con:concfinhitstr} is equivalent to the statement, that for all $k \ge 3$ we have $\sup \set{n(F) : F \in \Ipruclashi{\delta=k}} \le 4 k - 5$.
\end{corollary}

Before we can finally prove the main result of this \Schrift, we need two lemmas on nfs-reducibility. First we show that clause-irreducible clause-sets with a variable occurring positively and negatively exactly twice are nfs-reducible:
\begin{lemma}\label{lem:22fns}
  Consider $F \in \Uclash$ with $v \in \var(F)$ such that $\ldeg_F(v) = \ldeg_F(\ol{v}) = 2$. Then $F$ is fs-resolvable or allows an nfs-flip enabling an fs-resolution.
\end{lemma}
\begin{proof}
Assume that $F$ has no fs-pairs. Let $C_1, C_2$ be the two $v$-occurrences and let $D_1, D_2$ be the two $\ol{v}$-occurrences. There is a literal $x \in C_1$ with $\ol{x} \in C_2$. If $\var(x) \notin \var(D_1) \cup \var(D_2)$, then setting $x$ to true (or false, it doesn't matter) we create an UHIT where $v$ becomes singular, and via Corollary \ref{cor:2suhit} then $\set{D_1,D_2}$ is an fs-pair; thus $\var(x) \in \var(D_1) \cup \var(D_2)$. If $\var(x) \in \var(D_1) \cap \var(D_2)$, then w.l.o.g.\ $x \in D_1$, $\ol{x} \in D_2$ (if there wouldn't be a clash, then via setting $x$ to true resp.\ false one could create an UHIT with $v$ occurring only once); now setting $x$ to false (or true, again it doesn't matter) $v$ becomes 1-singular, and via Corollary \ref{cor:charac1sUHIT} $\set{C_1 \sm \set{x}, D_1 \sm \set{x}}$ is an fs-pair, whence $\set{C_1, D_1}$ would be an fs-pair; thus $\var(x) \notin \var(D_1) \cap \var(D_2)$. So finally w.l.o.g.\ $x \in D_1$, $\var(x) \notin \var(D_2)$. So after setting $x$ to true/false we can apply Corollary \ref{cor:charac1sUHIT} resp.\ \ref{cor:2suhit}, and we obtain that there is a clause $A$ and a new literal $z$ such that $C_2 = \set{v,\ol{x}} \addcup A \addcup \set{z}$, $D_2 = \set{\ol{v}} \addcup A \addcup \set{z}$, and $D_1 = \set{\ol{v},x} \addcup A \addcup \set{\ol{z}}$. Applying the nfs-flip to $C_2, D_2$, from $D_2$ we obtain $D_2' := \set{\ol{v}, x} \addcup A \addcup \set{z}$, and now $\set{D_1,D_2'}$ is an fs-pair. \qed
\end{proof}

Furthermore, if we can reach deficiency $1$ by assigning one variable, then via a series of nfs-flips we can create an fs-pair:
\begin{lemma}\label{lem:createfs}
  Consider $F \in \Uclash$ and $x \in \lit(F)$ such that assigning $x$ to true in $F$ yields a clause-set with deficiency $1$. Then via a series of nfs-flips on $F$ we can reach an element of $\Uclash$ with an fs-pair.
\end{lemma}
\begin{proof}
If $F$ has an fs-pair, we are done, and so assume that $F$ is not fs-resolvable. Let $F'$ be the clause-set obtained by assigning $x$ to true (so $F' \in \Uclashi{\delta=1}$); we do induction on $c(F')$. If $c(F') = 1$, then $F = \set{\set{x}, \set{\ol{x}}}$, and we are done, so assume $c(F') > 1$. Now $F'$ contains an fs-pair $\set{\set{y} \addcup C, \set{\ol{y}} \addcup C}$ (as was already shown in \cite{AhLi86}), and thus $F$ contains w.l.o.g.\ the nfs-pair $\set{\set{y} \addcup C \addcup \set{\ol x}, \set{\ol{y}} \addcup C}$. Performing the nfs-flip replaces these two clauses by $C \addcup \set{\ol x}, \set{\ol{y}} \addcup C \addcup \set{x}$, and so the new $F'$ has been reduced by one clause (while still in $\Uclashi{\delta=1}$). \qed
\end{proof}

\begin{theorem}\label{thm:nfsdecompo234}
  $\Ipruclashi{\delta=3} = \es$.
\end{theorem}
\begin{proof}
Consider $F \in \Pruclashi{\delta=3}$, and assume that $F$ is nfs-irreducible. Consider some $v \in \var(F)$ with minimal $\vdeg_F(v)$. So $\vdeg_F(v) \in \set{4,5}$ (using Corollary \ref{cor:uhit3fs} and \cite[Theorem 15]{KullmannZhao2011Bounds}), and by Lemma \ref{lem:22fns} we have indeed $\vdeg_F(v) = 5$. W.l.o.g.\ $\ldeg_F(v) = 3$, contradicting Lemma \ref{lem:createfs} with $x = v$. \qed
\end{proof}

By Theorem \ref{thm:equivconjnfs}:
\begin{corollary}\label{cor:NV3}
  $\maxnhitdef(3) = 4 \cdot 3 - 5 = 7$.
\end{corollary}

\section{Conclusion and outlook}
\label{sec:concl}

We proved the strong form of the Finiteness Conjecture (FC, Conjecture \ref{con:concfinhitstr}) for deficiency $k=3$ (Corollary \ref{cor:NV3}), and developed on the way new tools for understanding (hitting) clause-sets:
\begin{itemize}
\item Full subsumption (fs) resolution and full subsumption (fs) extension (Definition \ref{def:fullsubres}): new aspects of one of the oldest methods in propositional logic (at least since \cite{Boole1854Gesetzeb}).
\item Singular variables and the singularity index (Sections \ref{sec:singvar}, \ref{sec:uhitnsingvar}): simple variables and their elimination and introduction.
\item \emph{Clause-factors, clause-factorisations, irreducible clause-sets} (Section \ref{sec:redsclscl}): generalising fs-resolution and singular variables through a structural approach.
\item \emph{Nearly full subsumption (nfs) resolution, nfs-irreducible clause-sets} (Section \ref{sec:Subsumptionflips}): extending the reach of clause-factorisations.
\end{itemize}
The proof of Corollary \ref{cor:NV3} works by the general reduction to the nfs-irreducible case (Theorem \ref{thm:equivconjnfs}), where there are no such cases for deficiencies up to $3$ (Theorem \ref{thm:nfsdecompo234}). Future steps are the determination of $\Uclashnsi{\delta=3}$ and the proof of FC for $k=4$. We believe clause-irreducible clause-sets are a valuable tool, and a fundamental question here is about a kind of prime-factorisation of UHITs into clause-irreducible clause-sets.

\bibliographystyle{plainurl}

\newcommand{\noopsort}[1]{}

\end{document}